\documentclass{article}
\usepackage{microtype}%if unwanted, comment out or use option "draft"

\usepackage{booktabs} % For formal tables
\usepackage{amsmath}
\usepackage{amssymb}
\usepackage{amsthm}
\usepackage{amsfonts}
\usepackage{xspace}
\usepackage{algorithm}
\usepackage{algorithmicx}
\usepackage{algpseudocode}
\usepackage{color}
\usepackage[T1]{fontenc}
\usepackage{bm}
\usepackage{thmtools}
\usepackage{lmodern}

\title{Randomized Communication Without Network Knowledge}

\author{Artur Czumaj \and Peter Davies}
\newtheorem{theorem}{Theorem}
\newtheorem{corollary}[theorem]{Corollary}
\newtheorem{lemma}[theorem]{Lemma}

\newcommand{\nat}{\ensuremath{\mathbb{N}}\xspace}

\newcommand{\Prob}[1]{\mathbb{P}\left[#1\right]}
\newcommand{\ci}{\ensuremath{c_1}\xspace}
\newcommand{\cii}{\ensuremath{c_2}\xspace}
\newcommand{\ciii}{\ensuremath{c_3}\xspace}
\newcommand{\civ}{\ensuremath{c_4}\xspace}

\begin{document}

\maketitle

\begin{abstract}
	Radio networks are a long-studied model for distributed system of devices which communicate wirelessly. When these devices are mobile or have limited capabilities, the system is often best modeled by the ad-hoc variant, in which the devices do not know the structure of the network. A large body of work has been devoted to designing algorithms for the ad-hoc model, particularly for fundamental communications tasks such as broadcasting. Most of these algorithms, however, assume that devices have some network knowledge (usually bounds on the number of nodes in the network $n$, and the diameter $D$), which may not always be realistic in systems with weak devices or gradual deployment. Very little is known about what can be done when this information is not available.
	
	This is the issue we address in this work, by presenting the first \emph{randomized} broadcasting algorithms for \emph{blind} networks in which nodes have no prior knowledge whatsoever. We demonstrate that lack of parameter knowledge can be overcome at only a small increase in running time. Specifically, we show that in networks without collision detection, broadcast can be achieved in $O(D\log\frac nD\log^2\log\frac nD + \log^2 n)$ time, almost reaching the $\Omega(D\log\frac nD + \log^2 n)$ lower bound. We also give an algorithm for directed networks with collision detection, which requires only $O(D\log\frac nD\log\log\log\frac nD + \log^2 n)$ time.
\end{abstract}

\section{Introduction}

\subsection{Model and problem}
We study the classical model of \emph{multi-hop radio networks}.

\paragraph{Multi-hop radio networks.}
In this model, a communications network is represented as a (directed or undirected) graph, with nodes corresponding to devices with wireless capability. An edge $(u,v)$ in the graph means that device $u$ can reach device $v$ via direct transmission. Efficiency of algorithms is measured in terms of number of nodes $n$ in the network, and eccentricity $D$ (the distance between the furthest pair of nodes in the network).

The defining feature of radio networks is the rule for how nodes can communicate: time is divided into discrete synchronous steps, and in each step every node can choose whether to transmit a message or listen for messages. A listening node in a given time-step then hears a message iff exactly one of its in-neighbors transmits. In the model \emph{with collision detection}, a listening node can distinguish between the cases of having 0 in-neighbors transmit and having more than one, but in the model \emph{without collision detection} these scenarios are indistinguishable. We study both variants of the model.

\paragraph{Node knowledge.}
We are concerned with the \emph{ad-hoc} variant of the \emph{multi-hop radio network} model, which means that we assume nodes have no prior knowledge about network structure. However, it is usual for work on ad-hoc networks to assume that nodes do know the values of $n$ and $D$, or at least upper bounds thereof. We do not make this assumption, and thus are dealing with a more restrictive model, which we call \emph{blind} radio networks, in which nodes have no prior network knowledge whatsoever. We do assume that nodes have access to a \emph{global clock}, which tells them the absolute number of the current time-step. Our algorithm for the model \emph{without collision detection}, though, does not require this as an extra assumption, since nodes only participate once they have received the source message, and so a global clock can be simulated by appending the current time-step to the source message.

\paragraph{Task.}
We design \emph{randomized algorithms} for the task of \emph{broadcasting}. This is the most fundamental global communication task, in which a single designated \emph{source} node starts with a message, and must inform all nodes in the network via transmissions. We assume that all nodes except the source begin in an \emph{inactive} state (i.e. do not transmit), and become active when they are informed of the source message via a transmission from a neighbor.

Our algorithms will all be Monte-Carlo algorithms succeeding with \emph{high probability}. That is, we will give worst-case running times, and ensure that the failure probability is at most $n^{-c}$ for some $c>0$.

\subsection{Preliminaries.}
Since we are only concerned with the asymptotic performance of our algorithm, we will assume that $n$ is at least a sufficiently large constant throughout. To avoid negative terms when using logarithms, we will use $\log x$ to mean $\max\{\log_2 x,1\}$. We will use $\ci$, $\cii$, $\ciii$ \dots as sufficiently large constants whose value we will set at some point during the analysis.

\subsection{Related work}
Broadcasting is possibly the most studied problem in radio networks, and has a wealth of literature in various settings. In the most standard model of ad-hoc networks, that of networks {without collision detection}, the first major result was a seminal paper of Bar-Yehuda et al.\ \cite{-BGI92}, who designed an almost optimal randomized broadcasting algorithm achieving the running time of $O((D + \log n) \cdot \log n)$ with high probability. This bound was later improved by Czumaj and Rytter \cite{-CR06}, and independently Kowalski and Pelc \cite{-KP03b}, who gave randomized broadcasting algorithms that complete the task in $O(D \log \frac{n}{D} + \log^2 n)$ time with high probability. This running time matched a known $\Omega(D\log\frac nD + \log^2 n)$ lower bound for the task \cite{-ABLP91,-KM98}. All of these results hold for \emph{directed} networks as well as \emph{undirected} ones. 

More recently, Ghaffari, Haupler and Khabbazian \cite{-GHK13} showed that collision detection can be used to surpass this lower bound, attaining an $O(D+\log^6 n)$ time algorithm. Work by Haeupler and Wajc \cite{-HW16} demonstrated that even without collision detection, the lower bound could be beaten assuming \emph{spontaneous transmissions} were permitted; that is, nodes have access to a global clock and are allowed to transmit before receiving the source message. Czumaj and Davies \cite{-CD17} extended this approach and obtained a running time of $O(D\frac{\log n}{\log D}+\log^{O(1)} n)$ for the setting with spontaneous transmissions. However, these algorithms only work in undirected networks.

Deterministic algorithms for broadcasting have also been studied; for undirected networks the fastest known algorithm is the $O(n \log D)$-time algorithm of \cite{-K05}, while for directed networks it is the $O(n \log D\log\log D)$-time algorithm of \cite{-CD16}.

All of these results also \emph{intrinsically require parameter knowledge}, and algorithms that do not require such knowledge have been little studied. The closest analogue in the literature is the work of Jurdzinski and Stachowiak \cite{-JS05}, who give algorithms for wake-up in single-hop radio networks (those in which the underlying graph is a clique, i.e. $D=1$) under a wide range of node knowledge assumptions. Their Use-Factorial-Representation algorithm is the most relevant; the running time is given as $O((\log n\log\log n)^3)$ for high-probability wake-up with a global clock (a slightly stronger task than broadcasting) in single-hop networks, but a similar analysis as we present here would demonstrate that the algorithm also performs broadcasting in multi-hop networks in $O((D+\log n)\log^2 \frac nD \log^3\log \frac nD)$ time.

A deterministic algorithm for broadcasting in radio networks without parameter knowledge is given in \cite{-CD18a}, with a running time of $O(n\log L \log\log n)$, where $L$ is the range of unique IDs with which nodes are equipped.

\subsection{New results}
We present a randomized algorithm for broadcasting in (directed or undirected) networks without collision detection which succeeds with high probability within time $O(D\log \frac nD\log^2\log\frac nD + \log^2 n)$. This improves over the \sloppy{$O((D+\log n)\log^2 \frac nD \log^3\log \frac nD)$} time that could be obtained by applying our analysis method to the Use-Factorial-Representation algorithm of Jurdzinski and Stachowiak \cite{-JS05}, and comes within a poly-$\log\log$ factor of the $\Omega(D\log\frac nD + \log^2 n)$ lower bound.

We also present an algorithm for directed networks with collision detection, whose $O(D\log \frac nD\log\log\log\frac nD + \log^2 n)$ running time comes even closer to the lower bound (we note that, to the authors' knowledge, it has not been proven that the lower bound still holds in this setting, though it would be very surprising if it did not).

Finally, we make the observation that in undirected networks with collision detection, the $O(D+\log^6 n)$-time algorithm of \cite{-GHK13} can be simulated without parameter knowledge at no extra cost.

%===============================================================

\section{Algorithms}

\subsection{Outline of Algorithms and Analysis}

The main idea of our algorithms is as follows: when considering a particular node $v$ we wish to inform, all of its active in-neighbors will be transmitting with some probability. We wish to make the sum of these probabilities approximately constant (say $\frac 12$), since then we can use the following lemma (variants of which have been used in many previous works such as \cite{-CD16}) to show that $v$ will be informed with good probability:

\begin{lemma}\label{lem:colhit}
	Let $x_i$, $i\in [n]$, be independent $\{0,1\}$-valued random variables with $\Prob{x_i=1}\leq \frac 12$, and let $f =\sum_{i\in [n]} \Prob{x_i=1}$. Then $\Prob{\sum_{i\in [n]} x_i = 1} \geq f4^{-f}$.
\end{lemma}

\begin{proof}
	\begin{align*}
	\Prob{\sum_{i\in [n]} x_i = 1}
	&= \sum_{j\in [n]} \Prob{x_j = 1 \land x_i = 0 \forall i \neq j}
	\geq \sum_{j\in [n]} \Prob{x_j = 1}\cdot \Prob{x_i = 0 \forall i}\\
	&\geq f\cdot \Prob{x_i = 0 \forall i}
	= f\cdot \prod_{i\in [n]} (1- \Prob{x_i = 1})
	\geq f\cdot \prod_{i\in [n]} 4^{-\Prob{x_i = 1}} \\
	&= f\cdot 4^{-\sum_{i\in [n]} \Prob{x_i = 1}}
	=f4^{-f}
	\enspace.
	\qedhere
	\end{align*}
\end{proof}

However, we do not know the size of $v$'s active in-neighborhood, so choosing appropriate probabilities is difficult. To do so, we have the source node generate a global random variable for each time-step, which will function as a `guess' of in-neighborhood size. By appending these variables to the source message, we can ensure that all active nodes are aware of them. Then, based on these global variables and upon local randomness, the active nodes decide whether to transmit.

This is complicated by the fact that there is no `perfect' choice of transmission probabilities which works well for all network conditions. Our algorithms therefore require performing several different protocols simultaneously (using random time multiplexing). This results in a framework given by Algorithm \ref{alg:B0}.

\begin{algorithm}[H]
	\caption{Broadcast Framework}
	\label{alg:B0}
	\begin{algorithmic}
		\For {$t = 1$ to $\infty$}
		\State let $T=2^t$
		\State $s$ randomly generates a sequence $S\in [C]^{T}$ with independent uniformly chosen entries.
		\State for each $j\in[T]$, $s$ generates a random variable $x_j$ from distribution $Y_{S_j,T}$.
		\State $s$ appends $S$ and variables $x_j$ to the source message.
		\For {$j$ from $1$ to $T$, in time-step $j$,}
		\State active nodes $v$ transmit with probability $p_{S_j,T}(x_j)$.
		\EndFor	
		\State reset time-step numbers and set non-source nodes to inactive.
		\EndFor	
	\end{algorithmic}
\end{algorithm}

Here we have some constant number $C$ of different protocols, each of which are equipped with a distribution $Y$ which tells the source how to choose the global random variables, and a probability function $p$ which tells nodes how to use this (and parameter $T$) to determine their transmission probabilities. We are using $T$ as a doubling parameter which approximates the true time-step number; this is because the source cannot send an unbounded amount of information, and so must guess in advance how many time-steps to generate randomness for.

By analyzing these protocols we can obtain some bound on the probability that a node with active neighbors is informed, in each time-step. Then, the total amount of time we must wait to inform that node can be bounded by a geometric random variable. To sum this waiting time over all the nodes in a path, we can use the following lemma from \cite{-CR06} about the concentration of sums of independent geometric random variables:

\begin{lemma}[Lemma 3.5 of \cite{-CR06}]\label{lem:rv}
	
	Let $X_1,\dots,X_D$ be a sequence of independent integer-valued random variables, each $X_i$ geometrically distributed with parameter $p_i$, $0<p_i<1$. For every $i$, let $\mu_i = 1/p_i$, and let $M$ be the set of unique $\mu_i$, i.e. $M = \{\mu_i:1\leq i \leq D\}$. If $\sum_{i=1}^{d} \mu_i \leq N$, then for any positive real $\beta$,
	\[ \Prob{\sum_{i=1}^{D}X_i \leq 2\cdot N + 8 \ln (|M|/\beta)\cdot \sum_{z\in\Delta}z}\geq 1-\beta \]
\end{lemma}

\begin{corollary}\label{cor:rv}	
	Let $X_1,\dots,X_D$ be a sequence of independent integer-valued random variables, each $X_i$ geometrically distributed with parameter $1/\mu_i$, $\mu_i \in \nat$. Let $\mu_{max}$ be the maximum $\mu_i$. If $\sum_{i=1}^{D} \mu_i \leq N$, then for any positive $\beta \leq \frac {1}{\log \mu}$,
	\[ \Prob{\sum_{i=1}^{D}X_i \leq 4 N + 65 \mu\ln (1/\beta)}\geq 1-\beta \]
\end{corollary}

\begin{proof}
	Let $M$ be the set of all powers of $2$ up to $\mu$, i.e. $M = \{2^i:1\leq i\leq \lceil \log \mu \rceil\}$; then $|M|= \lceil\log \mu \rceil$ and $\sum_{z\in M}z\leq 4\mu$. For all $i$, let $X'_i$ be a geometric random variable with $\mu'_i$ equal to $\mu_i$ rounded up to the next power of $2$ (and $p'_i = 1/\mu'_i$ accordingly). Note that $X'_i$ majorizes $X_i$. Then, by Lemma \ref{lem:rv}, for positive $\beta \leq \frac {1}{\log \mu}$,
	
	\begin{align*}
	1-\beta &\leq \Prob{\sum_{i=1}^{D}X'_i \leq 2\cdot \sum_{i=1}^{D} \mu'_i + 8 \ln (|M|/\beta)\cdot \sum_{z\in M}z}\\
	&\leq \Prob{\sum_{i=1}^{D}X'_i \leq 4\cdot \sum_{i=1}^{d} \mu_i + 8 \ln (\lceil\log \mu \rceil/\beta)\cdot 4\mu}\\
	&\leq \Prob{\sum_{i=1}^{D}X_i \leq 4N + 65 \mu \ln (1/\beta)}\enspace.
	\end{align*}
\end{proof}

Using this bound, we have a recipe for getting from these geometric random variable to a running time for a broadcasting algorithm:

\begin{lemma}\label{lem:layeran}
	Let $\mu:[D]\times[n] \rightarrow [n]$ be a function which is non-decreasing in its second argument. Let $N$ be the maximum of $\sum_{d=1}^{D}\mu(d,\delta_d)$, subject to $\sum_{d=1}^{D} \delta_d \leq n$, and let $\mu_{max}$ be the maximum value of $\mu$. If an algorithm for broadcasting guarantees that any node $v$ at distance $d_v$ from the source with neighborhood of size $\delta_v$ is informed within time $X_v$ of a neighbor being informed, where $X_{v}$ is stochastically majorized by a geometric random variable with parameter $\frac{1}{\mu(d_v,\delta_v)}$, then probability at least $\frac {1}{n}$, broadcasting is completed in the whole network within $4 N + 91 \mu_{max}\log n$ time.
\end{lemma}

\begin{proof}
	Fix some arbitrary target node $u$ and some shortest $(s,u)$ path $p=(s = p_0, p_1, \dots, p_{d_v} = v)$. Classify nodes into layers as follows: let layer $L_i$ be the set of all nodes whose latest out-neighbor path node is $p_i$. We call a layer \emph{leading} if it is the furthest layer containing an active node. We wish to bound the time $t_i$ that a layer $L_i$ can remain leading, since the total time taken to inform $u$ is then at most $\sum_{i=1}^{D} t_i$. This time $t_i$ is stochastically majorized by geometric random variable with parameter $\frac{1}{\mu(i,|L_i|)}$, since only nodes in the intersection of $u$'s neighborhood and $L_i$ can participate in informing $u$ while $L_i$ is leading. Then, applying Corollary \ref{cor:rv}, 
	
	\[ \Prob{\sum_{i=1}^{D}t_i \leq 4 N + 65 \mu_{max}\ln (n^2)}\geq 1-n^{-2}\enspace, \]
	
	i.e. with probability at least $1-n^{-2}$, $u$ is informed within time $4 N + 91 \mu_{max}\log n$, and taking a union bound over all nodes $u$, the whole network is informed within this time with probability at least $1-n^{-1}$.
\end{proof}

In the next section we describe the choice of protocols with which to use this framework to achieve our stated running times.

\subsection{Networks Without Collision Detection}
In networks without collision detection, we employ two different protocols, i.e. $C=2$. We will call the protocol we use under most circumstances \textbf{General-Broadcast}; in this protocol, the source `guesses' a neighborhood size from $1$ to $\infty$ in each time-step, with a probability that decreases in neighborhood size in order for the total probability to sum to at most $1$. Transmission probabilities are independent of parameter $T$. If we used only this protocol, we would obtain a running time of $O((D+\log n) \log \frac nD\log^2\log \frac nD)$. In low diameter networks (when $D<\log n$), we improve upon this with \textbf{Shallow-Broadcast} protocol, which informs networks of low diameter in $O(\log^2 n)$ time by assuming that $T \approx \log^2 n$ and using this to approximate the maximum in-neighborhood size to account for.

\paragraph{Shallow-Broadcast}
Distribution $Y_{1,T}$ is given by $\Prob{x_j = y} = \frac{\ci}{\sqrt T}$ for all $y\in [\frac{\sqrt T}{\ci}]$. Probability function $p_{1,T}$ is given by $p_{1,T}(x_j)=2^{-x_j}$.

\begin{lemma}\label{lem:B11}
	If $D\leq \log n$, \textbf{Shallow-Broadcast} performs broadcasting in $O(\log^2 n)$ time with high probability.
\end{lemma}

\begin{proof}
	We consider the iteration in which $(\ci\log n)^2 \leq T \leq 2(\ci\log n)^2$. Fix a time-step $j$, and let $u$ be an inactive node with a set $\Delta$ of active neighbors, $\delta=|\Delta|\geq 1$. With probability $\frac 1C$, $j$ is a \textbf{Shallow-Broadcast} time-step. Then, with probability $\frac{\ci}{\sqrt T}\geq\frac{1}{\sqrt 2\log n}$, $x_j$ is chosen such that $2\delta\leq 2^{x_j}\leq 4\delta$, in which case $\frac 14\leq \sum_{u\in\Delta}\Prob{u\text{ transmits}}\leq \frac 12$, so by Lemma \ref{lem:colhit}, $\Prob{v\text{ is informed}}\geq \frac 14 \cdot 4^{-\frac 14} \geq \frac 16$. So, in each time-step, $u$ is informed with probability at least $\frac{1}{C}\cdot \frac{1}{\sqrt 2\log n} \cdot \frac 16 \geq \frac{1}{9C\log n}$. Time taken to inform $u$ is therefore stochastically majorized by a geometric random variable with parameter $\frac{1}{9C\log n}$. Using Lemma \ref{lem:layeran} with $\mu(d,\delta) = 9C\log n$ for all $d,\delta$, we can conclude that the network will be informed within $4 N + 91 \mu_{max}\log n \leq 36CD\log n + 819C\log^2 n$ time with high probability. We set $\ci = 30C$ to ensure that the iteration we analyze is sufficiently long, and thus broadcast is performed in $O(\log^2 n)$ time.
\end{proof}

\paragraph{General-broadcast}
Distribution $Y_{2,T}$ is given by $\Prob{x_j = y} = \frac{1}{3 y \log^2 y}$ for all $y\in \nat$ (and $x_j=0$ with the remaining probability). Probability function $p_{2,T}$ is given by $p_{2,T}(x_j)=2^{-x_j}$.

We first check that $Y_{2,T}$ is a well-defined probability distribution, which is the case since $\sum_{y\in \nat}	\frac{1}{3 y \log^2 y} \leq \frac 12+ \int_{2}^{\infty}\frac{dy}{3 y \log^2 y }\leq 1$.

\begin{lemma}\label{lem:B12}
	If $D\geq \log n$, \textbf{General-Broadcast} performs broadcasting in $O(D\log \frac nD\log^2\log\frac nD)$ time with high probability.
\end{lemma}

\begin{proof}
	We consider the iteration in which $\cii D\log \frac nD\log^2\log \frac nD \leq T \leq 2\cii  D\log \frac nD\log^2\log \frac nD$. Fix a time-step $j$, and let $u$ be an uninformed node with a set $\Delta$ of informed neighbors, $\delta=|\Delta|\geq 1$. With probability $\frac 1C$, $j$ is a \textbf{General-Broadcast} time-step. Then, with probability at least $\frac{1}{3 \log (4\delta) \log^2 \log (4\delta)}\geq \frac{1}{6\log\delta\log^2\log \delta}$, $x_j$ is chosen such that $2\delta\leq 2^{x_j}\leq 4\delta$, in which case $\frac 14\leq \sum_{u\in\Delta}\Prob{u\text{ transmits}}\leq \frac 12$, so by Lemma \ref{lem:colhit}, $\Prob{v\text{ is informed}}\geq \frac 14 \cdot 4^{-\frac 14} \geq \frac 16$. So, in each time-step, $u$ is informed with probability at least $\frac{1}{C}\cdot\frac{1}{6\log\delta\log^2\log \delta} \cdot \frac 16 \geq \frac{1}{36C\log\delta\log^2\log \delta}$. Time taken to inform $u$ is therefore stochastically majorized by a geometric random variable with parameter $\frac{1}{36C\log\delta\log^2\log \delta}$. Using Lemma \ref{lem:layeran} with $\mu(d,\delta) = \frac{1}{36C\log\delta\log^2\log \delta}$ for all $d,\delta$, we can conclude that the network will be informed within 
	\begin{align*}
	4 N + 91 \mu_{max}\log n &\leq 144CD\log\frac nD\log^2\log \frac nD + 3276\log^2 n\log^2\log n\\&\leq 3500CD\log \frac nD\log^2\log \frac nD
	\end{align*} time with high probability. We set $\cii = 3500C$ to ensure that the iteration we analyze is sufficiently long, and thus broadcast is performed in $O(D\log \frac nD\log^2\log\frac nD)$ time.
\end{proof}

To perform broadcasting in networks without collision detection, we apply the framework of Algorithm \ref{alg:B0} using the \textbf{Shallow-Broadcast} and \textbf{General-Broadcast} protocols.

\begin{theorem}
	Broadcasting can be performed in networks without collision detection in $O(D\log \frac nD\log^2\log\frac nD+\log^2 n)$ time, with high probability.
\end{theorem}

\begin{proof}
	Follows from Lemmas \ref{lem:B11} and \ref{lem:B12}.
\end{proof}

\subsection{Directed Networks With Collision Detection}
When collision detection (and a global clock) is available, nodes can determine their exact distance from the source node within $O(D)$ time, via a process known as \emph{beep-waves}: in time-step $1$, the source node emits a `beep' (a transmission with arbitrary content), and in every subsequent step, all nodes who hear either a transmission or a collision in a time-step themselves `beep' in the next time-step. In this way, the wave of beeps emanates out from the source, one distance hop per time-step, and the time-step number in which a node hears its first beep is equal to its distance from the source. For more detail of beep-wave techniques, see \cite{-CD15}.

If we perform this beep-wave procedure before our main algorithm, we can assume every node $v$ knows its distance $d_v$ from the source (actually, since we cannot tell when the procedure has ended, we must use time multiplexing, for example performing beep waves during odd time-steps and the main algorithm during even ones, but this does not affect asymptotic running time). The local transmission probabilities that nodes use during our broadcasting algorithm can then depend on $d_v$, as well as $T$ and the global randomness provided by the source. This is what we use to improve running time; we do not employ collision detection at any other point in the algorithm.

We add two new protocols to the two already defined, i.e. we now use $C=4$. The main new protocol is \textbf{Deep-Broadcast}, which assumes that $T\approx D\log\frac nD\log\log\log \frac nD$ and $d_v\approx D$, and uses this to approximate $(\frac nD)^2$, the largest neighborhood size for which it accounts. This only works well for nodes which do indeed have $d_v\approx D$ and $\delta_v\leq(\frac nD)^2$ , but for our analysis method this covers most nodes of importance, and we use \textbf{General-Broadcast} to deal with the remaining nodes. By including \textbf{Deep-Broadcast} we speed up broadcasting to $O(D\log\frac nD\log\log\log \frac nD)$ when $D>\log n\log^2\log n$, but when $D$ is below this, the running time of \textbf{General-Broadcast} still dominates. So, we also add \textbf{Semi-Shallow-Broadcast}, which works quickly for networks with $\log n \leq D\leq \log n\log^2\log n$.

\paragraph{Semi-Shallow-Broadcast}
Distribution $Y_{3,T}$ is given by \[\Prob{x_j = y}=\begin{cases}
\sqrt{\frac{\ciii ^2\log^2 T\log\log T}{2T}},&\text{ if } 1\leq y\leq \sqrt{\frac{T}{\ciii^2\log^2 T\log\log T}}\\
\frac{1}{3y\log\log T},&\text{ if } \sqrt{\frac{T}{\ciii^2\log^2 T\log\log T}}<y\leq \sqrt{\frac{T\log^2 T}{\ciii^2\log\log T}}\\
\end{cases}\enspace. \]
We let $x_j = 0$ with any remaining probability.

Probability function $p_{2,T,d_v}$ is given by $p_{2,T,d_v}(x_j)=2^{-x_j}$.

We first check that the distribution $Y_{3,T}$ is well defined, which is the case since 

\[\sum_{y = 1}^{\sqrt{\frac{T}{\ciii^2\log^2 T\log\log T}}}\sqrt{\frac{\ciii ^2\log^2 T\log\log T}{2T}} <\frac 34\enspace,\]

and

\[\sum_{y=\sqrt{\frac{T}{\ciii^2\log^2 T\log\log T}}+1}^{\sqrt{\frac{T\log^2 T}{\ciii^2\log\log T}}} \frac{1}{3y\log\log T} \leq \int\limits_{\sqrt{\frac{T}{\ciii^2\log^2 T\log\log T}}}^{\sqrt{\frac{T\log^2 T}{\ciii^2\log\log T}}} \frac{dy}{3y\log\log T} \leq \frac{\ln (\log^2 T)}{3\log\log T}<\frac 14\]
\begin{lemma}\label{lem:B21}
	If $\log n \leq D\leq \log n\log^2\log n$, \textbf{Semi-Shallow-Broadcast} performs broadcasting in $O(D\log n \log\log\log n)$ time with high probability.
\end{lemma}

\begin{proof}
	We consider the iteration in which $\ciii D\log n \log\log\log n \leq T \leq 2\ciii D\log n \log\log\log n$. Fix a time-step $j$, and let $u$ be an uninformed node with a set $\Delta$ of informed neighbors, $\delta=|\Delta|\geq 1$. With probability $\frac 1C$, $j$ is a \textbf{Semi-Shallow-Broadcast} time-step.  The probability that $x_j$ is chosen such that $2\delta\leq 2^{x_j}\leq 4\delta$ is at least
	
	\begin{align*}
	\sqrt{\frac{\ciii ^2\log^2 T\log\log T}{2T}} &\geq \sqrt{\frac{\ciii ^2\log^2 (2\ciii D\log n \log\log\log n)\log\log (2\ciii D\log n \log\log\log n)}{4\ciii D\log n \log\log\log n}}\\
	&\geq \sqrt{\frac{\ciii ^2\log^2 \log n\log\log \log n}{4\ciii D\log n \log\log\log n}}\\
	&\geq \sqrt{\frac{\ciii\log^2 \log n}{4 \log^2 n\log^2\log n}}
	= \frac{\sqrt{\ciii}}{2 \log n}\enspace,
	\end{align*}
	
	if $\log(4\delta)\leq \sqrt{\frac{T}{\ciii^2\log^2 T\log\log T}}$, and
	
	\begin{align*}
	\frac{1}{3\log(4\delta)\log\log T} \geq \frac{1}{3\log (4n)\log\log T}
	\geq \frac{1}{4 \log n \log\log\log n}\enspace,
	\end{align*}
	
	otherwise. Note that it is not possible that $	\log(4\delta)> \sqrt{\frac{T\log^2 T}{\ciii^2\log\log T}}$, since that would give
	
	\begin{align*}
	\log(4\delta) &\geq \sqrt{\frac{\ciii D\log n \log\log\log n \log^2 T}{\ciii^2\log\log (\ciii D\log n \log\log\log n)}}\\
	&\geq \sqrt{\frac{\ciii \log^2 n \log^2 \log n \log\log\log n}{4\ciii^2\log\log\log n}}\\
	&\geq\frac{\log n\log\log n}{2\sqrt\ciii}
	> \log(4n)\enspace,
	\end{align*}
	
	i.e. $\delta>n$, which is a contradiction.
	
	So, an appropriate value of $x_j$ is chosen with probability at least $\frac{1}{4 \log n \log\log\log n}$, in which case $\frac 14\leq \sum_{u\in\Delta}\Prob{u\text{ transmits}}\leq \frac 12$, so by Lemma \ref{lem:colhit}, $\Prob{v\text{ is informed}}\geq \frac 14 \cdot 4^{-\frac 14} \geq \frac 16$. Therefore, in each time-step, $u$ is informed with probability at least $\frac{1}{C}\cdot\frac{1}{4 \log n \log\log\log n} \cdot \frac 16 \geq \frac{1}{24C \log n \log\log\log n}$. Time taken to inform $u$ is therefore stochastically majorized by a geometric random variable with parameter $\frac{1}{24C \log n \log\log\log n}$. Using Lemma \ref{lem:layeran} with $\mu(d,\delta) = 24C \log n \log\log\log n$ for all $d,\delta$, we can conclude that the network will be informed within 
	\begin{align*}
	4 N + 91 \mu_{max}\log n &\leq 96CD \log n \log\log\log n + 2184C\log^2 n\log\log\log n\\&\leq 2280CD\log n \log\log\log n
	\end{align*} time with high probability. We set $\ciii = 2280C$ to ensure that the iteration we analyze is sufficiently long, and thus broadcast is performed in $O(D\log n \log\log\log n)$ time.
\end{proof}

\paragraph{Deep broadcast}
Distribution $Y_{4,T}$ is given by $\Prob{x_j = y}=\frac{1}{T}$ for all $y\in [T]$.
Probability function $p_{4,T,d_v}$ is given by $p_{4,T,d_v}(x_j)=2^{-\frac{x_i}{\civ d_v\log\log \frac {T}{d_v}}}$.

\begin{lemma}\label{lem:B22}
	If $D\geq \log n(\log\log n)^2$, \textbf{Deep-Broadcast} and \textbf{General-Broadcast} together complete broadcasting in $O(D\log \frac nD \log\log\log \frac nD)$ time, with high probability.
\end{lemma}

\begin{proof}
	We consider the iteration in which \[\civ^2 D\log \frac nD \log\log\log \frac nD \leq T \leq 2\civ^2 D\log \frac nD \log\log\log \frac nD\enspace.\] Fix a time-step $j$, and let $v$ be an uninformed node distance with a set $\Delta$ of informed neighbors, $\delta=|\Delta|$. Denote by $d$ $v$'s distance from the source.
	
	If $d<\frac{D}{\log^2\log\frac nD}$ or $\delta>(\frac nD)^2$, we analyze \textbf{General-Broadcast} time-steps, and conclude, as in the proof of Lemma \ref{lem:B12}, that the time taken to inform $v$ is stochastically majorized by a geometric random variable with parameter $\frac{1}{36C\log\delta\log^2\log \delta}$.
	
	Otherwise, we analyze \textbf{Deep-Broadcast} time-steps:
	
	There is some real value $x'\in [1,T]$ such that $\sum_{u\in\Delta}p_{4,T,d_u}(x') = \frac 12$, since the value of this sum is continuous in $x$, is at least 
	\[\sum_{u\in \Delta}2^{-\frac{1}{\civ d_u\log\log \frac {T}{d_u}}}\geq 2^{-\frac{1}{\civ (d-1)\log\log \frac {T}{d-1}}}> 2^{-1}=\frac 12\enspace,\]
	when $x=1$, and is at most 
	\begin{align*}
	\sum_{u\in \Delta}2^{-\frac{T}{\civ d_u\log\log \frac {T}{d_u}}}&\leq \sum_{u\in\Delta}2^{-\frac{\civ^2 D\log\frac nD \log\log\log \frac nD}{D\log\log \frac {\civ^2 D\log\frac nD \log\log\log \frac nD}{D}}}\\
	&\leq  \sum_{u\in \Delta}2^{-\frac{\civ^2 D\log\frac nD \log\log\log \frac nD}{2D\log\log \log\frac nD }}\\
	&\leq(\frac nD)^2\cdot  2^{-\frac 12 \civ^2\log\frac nD}= 2^{(2-\frac 12 \civ^2)\log\frac nD} \\&\leq 2^{-1} = \frac 12 \enspace,
	\end{align*} 
	
	when $x = T$.
	
	Then the value of the sum at $x'+\civ d\log\log\log \frac nD+1$ is at least 
	
	\begin{align*}
	\sum_{u\in \Delta}2^{-\frac{x'+\civ d\log\log\log \frac nD+1}{\civ d_u\log\log \frac {T}{d_u}}}&\geq\sum_{u\in \Delta}2^{-\frac{x'}{\civ d\log\log \frac {T}{d}}}\cdot 2^{-\frac{\civ d\log\log\log \frac nD+1}{\civ (d-1)\log\log \frac {T}{d-1}}}\\
	&\geq  2^{-1-\frac{d\log\log\log \frac nD+1}{(d-1)\log\log \frac {T}{d-1}}}\\
	&\geq  2^{-1-\frac{2d\log\log\log \frac nD}{d\log\log \log\frac nD}} =\frac 18\enspace.
	\end{align*} 
	
	If $x_j$ is chosen to be any of the integer values between $x'$ and $x'+\civ d\log\log\log \frac nD+1$, then by Lemma \ref{lem:colhit}, the probability that $v$ is informed is at least $\frac 18 \cdot 4^{-\frac 18} \geq \frac{1}{10}$. So, the overall probability that $v$ is informed in time-step $j$ is at least 
	\[\frac 1C \cdot \frac{1}{10} \cdot \frac{\civ d\log\log\log \frac nD}{T} \geq \frac{\civ d\log\log\log \frac nD}{20C\civ^2 D\log \frac nD \log\log\log \frac nD} = \frac{d}{20C\civ D\log \frac nD}\enspace.\]
	
	We are now in a position to apply Lemma \ref{lem:layeran} with \[\mu(d,\delta)=\begin{cases}
	\frac{20C\civ D\log \frac nD}{d},&\text{ if } d\geq\frac{D}{\log^2\log\frac nD} \text{ and } \delta\leq(\frac nD)^2\\
	36C\log\delta\log^2\log \delta)&\text{ otherwise. }
	\end{cases}
	\]
	
	We can bound $N$, the maximum of $\sum_{d=1}^{D}\mu(d,\delta_d)$, subject to $\sum_{d=1}^{D} \delta_d \leq n$, as follows: the maximum contribution of terms falling under the first case is at most
	
	\begin{align*}
	\sum_{d=\frac{D}{\log^2\log\frac nD}}^{D} \frac{20C\civ D\log \frac nD}{d}&\leq 20C\civ D\log \frac nD(1+\int_{\frac{D}{\log^2\log\frac nD}}^{D} \frac{di}{i})\\
	&\leq 20C\civ D\log \frac nD (1+\ln(\log^2\log\frac nD))\\
	&\leq 48C\civ D\log \frac nD\log\log\frac nD\enspace,
	\end{align*}
	
	the maximum contribution of terms where $d<\frac{D}{\log^2\log\frac nD}$ is at most 
	\begin{align*}
	\frac{D}{\log^2\log\frac nD} \cdot 36C\log\frac{n\log^2\log\frac nD}{D}\log^2\log \frac{n\log^2\log\frac nD}{D}&\leq \frac{37CD\log\frac nD\log^2\log \frac nD}{\log^2\log\frac nD}\\
	&=37CD\log\frac nD	\enspace,
	\end{align*}
	and the maximum contribution of terms where $\delta\geq(\frac nD)^2$ is at most
	\begin{align*}
	\frac{D^2}{n} \cdot 36C\log(\frac nD)^2\log^2\log (\frac nD)^2&\leq  73CD\cdot \frac Dn\log\frac nD \log^2\log\frac nD\leq 73CD\enspace.
	\end{align*}
	
	So, summing these contributions,
	
	\begin{align*}
	N&\leq 48C\civ D\log \frac nD\log\log\frac nD+37CD\log\frac nD+73CD
	\leq 50C\civ D\log \frac nD\log\log\frac nD\enspace.
	\end{align*}
	
	We can also see that the maximum $\mu$ value $\mu_{max}$ is at most $20C\civ\log n \log^2\log n$. So, applying Lemma \ref{lem:layeran}, we can conclude that the network will be informed within 
	\begin{align*}
	4 N + 91 \mu_{max}\log n &\leq 200C\civ D\log \frac nD \log\log\log \frac nD + 1820C\civ\log^2 n \log^2\log n\\
	&\leq 200C\civ D\log \frac nD \log\log\log \frac nD + 1820C\civ \cdot 2D\log \frac nD\\
	&\leq 3840C\civ D\log \frac nD \log\log\log \frac nD
	\end{align*} time with high probability. Here in the second inequality, we are using that $D\geq \log n(\log\log n)^2$. We set $\civ = 3840C$ to ensure that the iteration we analyze is sufficiently long, and thus broadcast is performed in $O(D\log \frac nD \log\log\log \frac nD)$ time.
	
\end{proof}
\begin{theorem}
	Broadcasting can be performed in networks with collision detection in $O(D\log \frac nD \log\log\log \frac nD + \log^2 n)$ time, with high probability.
\end{theorem}

\begin{proof}
	We perform a beep-wave to ensure that all nodes $v$ know their distance $d_v$ from the source, and then apply the framework of Algorithm \ref{alg:B0} using \textbf{Shallow-Broadcast}, \textbf{General-Broadcast}, \textbf{Semi-Shallow-Broadcast}, and \textbf{Deep-Broadcast}, i.e. $C=4$. If $D\leq \log n$, then by Lemma \ref{lem:B11} we complete broadcasting in $O(\log^2 n)$ time. If $\log n \leq D\leq \log n\log^2\log n$, then by Lemma \ref{lem:B21} we complete broadcasting in $O(D\log n \log\log\log n) = O(D\log \frac nD \log\log\log \frac nD)$ time. If $D > \log n\log^2\log n$, then by Lemma \ref{lem:B22} we complete broadcasting in $O(D\log \frac nD \log\log\log \frac nD)$ time. This gives a total asymptotic running time of $O(D\log \frac nD \log\log\log \frac nD + \log^2 n)$.
\end{proof}

\subsection{Undirected networks with collision detection}
In undirected networks in which collision detection is available, the fastest algorithm with known network parameters is the $O(D+\log^6 n)$-time result of \cite{-GHK13}. This algorithm involves utilizing beep-wave type methods to set up a structure known as a \emph{gathering spanning tree}, which arose in work on known-topology radio networks \cite{-GPX05} and admits a fast broadcasting schedule atop it. To achieve the $O(D+\log^6 n)$-time bound, constant-factor upper bounds on $D$ and $\log n$ are required. However, since the running time does not contain a product of these two quantities, the algorithm can be simulated without parameter knowledge by using a doubling parameter $T$, where $T$ is used as an upper bound for both $D$ and $\log^6 n$, and the algorithm is run for $T$ time in each iteration. Then, when $T$ exceeds both $D$ and $\log^6 n$, which happens within $O(D+\log^6 n)$ time, the algorithm will succeed.

\section{Conclusions}
We have presented the first randomized broadcasting algorithms for \emph{blind} radio networks. Since this is a new angle of research in radio networks, there are several interesting avenues for further research. Firstly, one would like to close the gap between the $O(D\log\frac nD\log^2\log\frac nD + \log^2 n)$ running time presented here for networks without collision detection and the $\Omega(D\log\frac nD + \log^2 n)$ lower bound. An improved lower bound would be of particular interest, since it would demonstrate that parameter knowledge does admit faster algorithms. The same applies to directed networks with collision detection, though here the first step would be to prove whether the $\Omega(D\log\frac nD + \log^2 n)$ lower bound still holds in this model.

A second possible extension to our work here is to achieve \emph{acknowledged} broadcasting, i.e. to ensure that by some time-step, all nodes know that broadcasting has been successfully completed and can cease transmissions. This would solve the issue with our algorithms here that nodes must continue transmissions indefinitely. However, it is not clear if acknowledged broadcasting is possible in this model; Chlebus et al. \cite{-CGGPR00} show that, using deterministic algorithms and without collision detection, it is not.

\newcommand{\Proc}{Proceedings of the\xspace}
\newcommand{\STOC}{Annual ACM Symposium on Theory of Computing (STOC)}
\newcommand{\FOCS}{IEEE Symposium on Foundations of Computer Science (FOCS)}
\newcommand{\SODA}{Annual ACM-SIAM Symposium on Discrete Algorithms (SODA)}
\newcommand{\COCOON}{Annual International Computing Combinatorics Conference (COCOON)}
\newcommand{\DISC}{International Symposium on Distributed Computing (DISC)}
\newcommand{\ESA}{Annual European Symposium on Algorithms (ESA)}
\newcommand{\ICALP}{Annual International Colloquium on Automata, Languages and Programming (ICALP)}
\newcommand{\IPL}{Information Processing Letters}
\newcommand{\JACM}{Journal of the ACM}
\newcommand{\JALGORITHMS}{Journal of Algorithms}
\newcommand{\JCSS}{Journal of Computer and System Sciences}
\newcommand{\PODC}{Annual ACM Symposium on Principles of Distributed Computing (PODC)}
\newcommand{\SICOMP}{SIAM Journal on Computing}
\newcommand{\SPAA}{Annual ACM Symposium on Parallelism in Algorithms and Architectures (SPAA)}
\newcommand{\STACS}{Annual Symposium on Theoretical Aspects of Computer Science (STACS)}
\newcommand{\TALG}{ACM Transactions on Algorithms}
\newcommand{\TCS}{Theoretical Computer Science}

\bibliographystyle{plain}

%%
%% Bibliography
%%

%% Please use bibtex, 

\bibliography{nk}

\end{document}